\documentclass[a4paper,11pt]{article}
\pdfoutput=1 

\usepackage{jheppub}

\usepackage[T1]{fontenc} 

\usepackage{amsthm}
\theoremstyle{definition}
\newtheorem*{hypH}{Hypothesis H}
\newtheorem{lemma}{Lemma}

\newtheorem{remark}{Remark}
\newtheorem*{corollary}{Corollary}

\usepackage{ifpdf}
\ifpdf
  \usepackage[all,pdf]{xy} 
\else
  \input xy 
  \xyoption{all}
  \xyoption{v2}
\fi

\SelectTips{cm}{11} 

\def \RR {\mathbb{R}}
\def \ZZ {\mathbb{Z}}
\def \TT {\mathbb{T}}
\def \into {\hookrightarrow}
\DeclareMathOperator{\String}{String}
\DeclareMathOperator{\Spin}{Spin}
\DeclareMathOperator{\Sp}{Sp}
\DeclareMathOperator{\pr}{pr}
\DeclareMathOperator{\ev}{ev}
\DeclareMathOperator{\Hom}{Hom}

\title{\boldmath Topological sectors for heterotic M5-brane charges under Hypothesis H}

\author{David Michael Roberts}
\affiliation{School of Mathematical Sciences,\\The University of Adelaide, Australia}
\emailAdd{david.roberts@adelaide.edu.au}

\abstract{Assuming Fiorenza--Sati--Schreiber's Hypothesis H, on the charge quantization of M-theory's $C$-field, the topological sectors of the resulting $\String^{c_2}(4)$-valued higher gauge theory on a heterotic M5-brane are classified by homotopy classes of maps from the worldvolume $\Sigma_{M5}$ to $B\String^{c_2}(4)$.
This note calculates the sectors in a number of examples of M5-brane topology, including examples considered in the 3d-3d correspondence, the emergence of skyrmions from higher-dimensional instantons and Witten's analysis of the S-duality of 4d Yang--Mills theory.}

\arxivnumber{2003.09832}

\notoc
\begin{document} 
\maketitle
\flushbottom

\section{Introduction}

Non-abelian higher gauge theory is having an increasing impact on string and M-theory \cite{Durham}.
Whereas the interpretation of the $B$-field as the curving on a gerbe is well-established \cite{FreedWitten,CMJ}, higher gauge theories generalising non-abelian gauge theories like Yang--Mills theory are less developed from the point of view of applications and examples, even if there is plenty of theoretical development (for example, but not limited to, \cite{FSS12,NSS15}).
Moreover, the appearance of higher non-abelian gauge theory can be quite unfamiliar, where the `fields' are much more exotic objects than just Lie algebra-valued forms, or sections of certain vector bundles (even the more familar RR fields are, when fully analysed \cite{DFM}, really valued in differential K-theory, a generalised cohomology theory \cite{Freed}).
This means it may be less obvious what the content of the gauge theory is, and what can be calculated about it.

Consider for the sake of analogy, Yang--Mills theory in 4D. 
If one is considering just topological information, as in this article, then we can ignore questions of metric signature.
The topological sectors for (classical) Yang--Mills with structure group $G$ on spacetime $X$ are classified by continuous maps $c\colon X\to BG$, where $BG$ is the classifying space of $G$.
More precisely, the sectors are in bijection with homotopy classes of such maps, where two maps $c_0$, $c_1$ are homotopic if there is a continuous family of maps $c_t\colon X\to BG$, $0\leq t \leq 1$.
The topologically trivial sector is represented by the map defined by $c(x) = \ast$ for all $x\in X$, where $\ast\in BG$ is a fixed basepoint.

In the case of the BPST $SU(2)$-instanton on $\RR^4$, there is an additional constraint, namely that as $|x| \to \infty$, $c_{\mathrm{BPST}}(x) \to \ast$.
Since this is purely topological, the rate of convergence is not important, unlike the case of thinking of the instanton as a gauge field, and requiring its energy to be finite.
As a consequence, $c_{\mathrm{BPST}}$ extends to a continuous map on the compactification of $\RR^4$, namely $S^4 \to BSU(2)$, sending $\infty \mapsto \ast \in BSU(2)$.
However, we know that homotopy classes of maps $S^4 \to BSU(2)$ are in bijection with the homotopy group $\pi_4(BSU(2)) \simeq \pi_3(SU(2)) \simeq \ZZ$.
Thus the topological sectors are labelled by instanton charge, as is well-known.

For higher gauge theory, the Lie group $G$ is replaced by an \emph{Lie $n$-group}: a higher categorical geometric object that has grouplike structure. 
The case of $n=2$ can be specified by giving a Lie crossed module (see the review in \cite{HDA_V}), with the most famous being the String 2-groups $\String_G$, for $G$ a compact simple simply-connected Lie group \cite{BCSS} (for $G=\Spin(n)$, we write $\String(n)$ instead of $\String_{\Spin(n)}$).
We will only be considering 2-groups in this article, and do not need the specifics of such structures; all that is needed is the fact that $2$-groups also have a classifying space \cite{BS}, analogous to $BG$.
For the topological sector analysis we present here, knowing the classifying space is sufficient, as we need to calculate the homotopy classes of maps to $B\String^{c_2}(4)$.

The 2-group we are interested in is a modified version of $\String(4)$, analogous in one sense to how $\Spin^c(n)$ is a modified version of $\Spin(n)$ \cite{Sati_11}. 
Whereas $\Spin(n) \to SO(n)$ is onto with kernel $\{\pm 1\}$, $\Spin^c(n) \to SO(n)$ is onto with kernel the larger group $U(1)$.
In the same way, $\String(n) \to \Spin(n)$ is onto with kernel the 2-group $\mathbf{B}U(1)$, but $\String^{c_2}(n)$ is onto with kernel the larger 2-group $\String_{\Sp(1)}$.
See \cite[\S 2.2.2]{SSS} for an extended discussion of the formal definition of these twisted String 2-groups, and their relevance to M-theory.

The reason that this family of 2-groups, or specifically the 2-group $\String^{c_2}(4)$, is of interest, is that from Hypothesis H \cite{Sati_18,FSS_20a} it follows that a heterotic M5-brane \cite{Ovrut} automatically carries fields from the higher gauge theory with $\String^{c_2}(4)$ as the structure 2-group \cite{FSS_20b}.
Hypothesis~H itself is about the true nature of the higher gauge theory of the $C$-field in M-theory:

\begin{hypH}[Fiorenza--Sati--Schreiber]
The M-Theory $C$-field is charge-quantized in $J$-twisted cohomotopy theory.
\end{hypH}

This purely mathematical hypothesis has been shown \cite{BSS19,SS19a,SS19b,FSS19c,FSS_20a,FSS_20b,SS20} to imply a host of anomaly cancellation and other consistency conditions previously proposed in the literature on physical grounds.

\begin{corollary}[{\cite[Theorem 1]{FSS_20b}}]
A heterotic M5-brane worldvolume $\Sigma_{M5}$ carries a $\String^{c_2}(4)$-valued higher gauge theory, topologically classified by a homotopy class of maps $\Sigma_{M5} \to B\String^{c_2}(4)$.
\end{corollary}

So while we do not work directly with the content of Hypothesis~H here, we will calculate the possible  topological sectors of the $\String^{c_2}(4)$ higher gauge theory on the simplest examples of M5-branes.
We note \cite{Sati_11} previously showed that M5-brane worldvolumes admitted certain twisted String structures, different to the ones considered here.

\begin{remark}
In what follows, we do not consider the quantum version of such a higher gauge theory, just as the analysis of the topological sectors of YM theory above merely considers the classical field theory.
\end{remark}

Just as in the case of ordinary gauge theory, the topological sectors for $\String^{c_2}(4)$ higher gauge theory on a manifold $M$ are given by homotopy classes of maps from $M$ to the classifying space $B\String^{c_2}(4)$.
To do this, we need to understand the homotopy type of the classifying space $B\String^{c_2}(4)$, at least in low dimensions.
This will be considered in Section~\ref{sect:pi_n_BStringc} below, before the calculations for each case of M5-brane topology in Section~\ref{sect:calc_sectors}.

It is worth outlining exactly how this higher gauge theory arises, since this is reflected in the construction of its classifying space, and in the classication calculations below.
We shall only give the barest outline, as the proof is rather involved, with many geometric objects coming together in exceptional ways.

Following \cite{FSS_20b}, we assume the spacetime manifold is of the form $\RR^{2,1}\times X^8$, for some 8-manifold $X^8$.
There is an embedding map $\Sigma_{M5} \xrightarrow{\phi} \RR^{2,1}\times X^8$ for the M5-brane.
The spacetime manifold carries the $C$-field, living in $J$-twisted cohomotopy, which can be represented by a map $C\colon \RR^{2,1} \times X^8 \to S^4 /\!\!/ \Spin(5) \simeq B\Spin(4)$.
By composition, this gives a map $\Sigma_{M5} \xrightarrow{C\circ \phi}  B\Spin(4)$.
The M5-brane carries, among other things, an $\Sp(1)$-valued gauge field, represented by a map $\Sigma_{M5} \to B\Sp(1)$.
Together, these satisfy an compatibility condition \cite{Sati_11}, so that the two maps, $\Sigma_{M5} \to B\Spin(4)$ and $\Sigma_{M5} \to B\Sp(1)$, assemble into a \emph{higher} gauge field $\Sigma_{M5} \to B\String^{c_2}(4)$.
Thus one can think of the $\String^{c_2}(4)$ higher gauge theory on the M5-brane as having a component inherited from the ambient spacetime, as well as a component arising from a gauge field living on the brane itself.
We shall discuss below for specific examples how to view the classification of the sectors in Section~\ref{subsect:class_table}, in light of this picture.

\section{Results}

The results for each M5-brane topology are given in Table~\ref{fig:results}, with a reference to where it is calculated in Section~\ref{sect:calc_sectors}. 
In the next subsection, we give some (mildly speculative) discussion about how these classifications relate to the physics literature, in a sample of cases.

\subsection{Sector classification}
\label{subsect:class_table}

Some of the cases in Table~\ref{fig:results} have boundary conditions, and this imposes extra constraints on the solutions, so that sectors can appear or disappear, even though the brane topology is identical.
We include the signature of the metric on the worldvolumes listed in the table, to keep things clear, but nothing relying on the metric will enter into the calculations.

\begin{table}[t]
\centering
\centerline{\begin{tabular}{cccc}
$\Sigma_{M5}$  & Boundary conditions & Sectors classified by & Details\\
\hline \\
$\RR^{5,1}$ & -- & $0$ & \ref{subsect:trivial} \\
$\RR^{5,1}$  & trivial at spatial $\infty$ & $(\ZZ_2)^3$ & \ref{subsect:flat_decay_at_infty} \\
$\RR^{4,1}\times S^1$ & -- & $0$ & \ref{subsect:trivial} \\
$\RR^{4,1}\times  S^1$ & trivial at spatial $\infty$ & $(\ZZ_2)^3$ & \ref{subsect:S1_decay} \\
$\RR^{3,1}\times S^2$ & -- & $0$ & \ref{subsect:trivial} \\
$\RR^{3,1}\times \TT^2$ & -- & $0$ & \ref{subsect:trivial} \\
$S^4 \times \TT^2$ & Wick rotated & $\ZZ^2 \times (\ZZ_2)^9$ & \ref{subsect:T2} \\
$\RR^{2,1} \times S^3$ & -- & $0$ & \ref{subsect:S3_unwrapped} \\
$\RR^{2,1} \times S^3$ & trivial at spatial $\infty$ & $(\ZZ_2)^3$ & \ref{subsect:S3_unwrapped_decay} \\
$\RR^3 \times S^3$ & Wick rot., trivial at $\infty$ & $(\ZZ_2)^3$ & \ref{subsect:S3xS3_Wick}\\
$\RR^{1,1}\times S^4$  & --  & $\ZZ^2$ & \ref{subsect:S4} \\
$\RR^{0,1}\times S^5$  & -- & $(\ZZ_2)^3$ & \ref{subsect:flat_decay_at_infty} \\
$\RR^{0,1} \times S^3 \times \TT^2$ & -- & $\ZZ^4 \times (\ZZ_2)^3$ & \ref{subsect:RxS3xT2} \\
\\
\hline
\end{tabular}}
\caption{Classification of topological sectors for a $\String^{c_2}(4)$ higher gauge theory on heterotic M5-brane worldvolume $\Sigma_{M5}$.}
\label{fig:results}
\end{table}

\subsection{Discussion}

Some of the above topologies for $\Sigma_{M5}$ are rather simple, and do no support nontrivial higher gauge sectors.
This is worth checking, just in case there are hidden surprises, or to confirm physical intution that there may not be anything topologially nontrivial, or to avoid fruitless searching for phenomenology when it cannot appear.
But some of the cases carry many interesting essentially independent topological sectors, and even torsion sectors, for instance the case of 5d instantons (the second and last rows of Table~\ref{fig:results}).
Even more interesting is when these nontrivial results arise on topologies that are already considered in the literature.
Only a few will be pointed out here; the reader can check their own favourite model against the table.

\begin{itemize}
\item The case of $\Sigma_{M5} = \RR^{2,1} \times S^3$ appears in the 3d-3d correspondence (see \cite{Dimofte} for a review).
Here we find that what is needed for nontrivial sectors is some spatial boundary condition, so that the gauge field extends to the (spatially) compactified $S^2\times S^3$.
A similar case to consider (after Wick rotation) is $\Sigma_{M5} = S^3 \times S^3$.
This is a more complicated calculation without some boundary condition on one of those $S^3$-factors, involving the homotopy type of a space of arbitrary maps on the 3-sphere.

\item Wick-rotated instantons on $\RR^{3,1}\times \TT^2$ and
Lorentzian instantons on $\RR^{4,1}\times S^1$ become skyrmions down in 4d (see eg, \cite{Sutcliffe}, highlighting ideas from \cite{AtiyahManton}). 
These latter are solitonic---and hence topologically stable---models for nucleons.
In the style presented here, the connection to instantons/skyrmions is somewhat obscured. 
However, by Remark~\ref{rem:homotopy_pullback} below, the homotopical data of a map to $B\String^{c_2}(4)$ consists of a map to $B\Spin(4)$ and a map to $B\Sp(1)\simeq BSU(2)$ and consistency data.
These (homotopy classes of) maps to $BSU(2)$ are more easily recognised as picking an instanton sector.
Further, given periodicity conditions, as in some of the calculations below, this results in (homotopy classes of) maps to $SU(2)$ itself, hence picking out a skyrmion sector.

\item The case of $\Sigma_{M5} = \RR^{0,1} \times S^3 \times \TT^2$ was considered in \cite{Witten_04}, who outlined how S-duality of 4d YM-theory should arise as the residual $SL(2,\mathbb{Z})$-action on $\TT^2$-compactifications of an M5-brane. 
Here $\TT^2$ is really (when the geometry is put back in) an elliptic curve (over $\mathbb{C}$), but topologically is just a torus.
Interestingly, Witten expresses a ``doubt'' in \cite{Witten_04} that the 6d ``quantum nonabelian gerbe theory'' he considers (on a stack of \emph{coincident} M5-branes) is the quantisation of a classical system; what we are considering here is a \emph{single} M5-brane, and the (nonabelian) higher gauge theory arises in a completely different way.

\end{itemize}

While we are not going to discuss each topology in detail, it might be worth considering exactly what these classes of maps are, in an illustrative case.
Unlike more familiar cases of classifications of instantons by integers, such as the 4d Yang--Mills instanton outlined above, we find that the higher gauge theory on $\RR^{1,1}\times S^4$ with no boundary conditions has sectors labelled by a \emph{pair} of integers.
These two copies of $\ZZ$ play different, asymmetric roles.
Recall that up to homotopy, the gauge field on $S^4$ (as we can ignore the $\RR^{1,1}$-factor) is given by a map $S^4 \to B\String^{c_2}(4)$.
That is, we are wrapping an $S^4$ inside the classifying space, which as in Section~\ref{sect:pi_n_BStringc} below, is constructed from a $B\Spin(4) \simeq B\Sp(1)_L\times B\Sp(1)_R$ and another $B\Sp(1)$, which we shall temporary denote by $\Sp(1)_G$.
The former is linked the geometry of the ambient spacetime and the latter to a gauge field on the M5-brane.
Recall first that $\pi_4(B\Sp(1)) = \ZZ$, so that a 4-sphere can wrap around a copy of $B\Sp(1)$ an integer number of ways, up to homotopy.
Because of the compatibility condition that gives us maps to $B\String^{c_2}(4)$, we can think of the sectors in this case as being labelled by the pair of integers $(n_L,n_R)$, where the $S^4$ wraps around the $B\Sp(1)_L$ $n_L$ times, around the $B\Sp(1)_R$ $n_R$ times, and the $B\Sp(1)_G$ $n_L$ times.

The appearance of torsion classes in, for example, the case $\Sigma_{M5} = S^5$, is also less familiar. 
Here the interpretation of such classes is perhaps a bit more speculative.
One option might be an old observation of Witten \cite{Witten_83}, that links the classes in $\pi_4(SU(2))\simeq \ZZ_2$ with whether a quantised soliton is a boson or a fermion.
Since homotopy classes of maps $S^5\to B\String^{c_2}(4)$ are assembled (similar to the previous paragraph) from compatible maps $S^5 \to B\Spin(4)$ and $S^5\to BSp(1)_G = BSU(2)$, the ordinary $SU(2)$ gauge field on $S^5$ lives in one of two sectors, corresponding to $\pi_5(BSU(2)) \simeq \pi_4(SU(2))$.
This can be thought of as an analogue of the instanton/skyrmion picture, where the instanton charge becomes the baryon number.
There is still, of course, two more factors of $\ZZ_2$, but these arise from spacetime geometry, which we are not considering here; their interpretation in this boson/fermion picture remains to be seen.

While the analysis here gives all possible sectors for a $\String^{c_2}(4)$ higher gauge theory on $\Sigma_{M5}$, it ignores the complete picture of \cite{FSS_20b}, whereby there will be additional aspects relating to the topology of spacetime itself, in which $\Sigma_{M5}$ sits.
One may visualise the situation as follows.
Fix an embedding $\Sigma_{M5} \xrightarrow{\phi} \RR^{2,1}\times X^8$.
Then what is analysed here is only the \emph{codomain} of the forgetful map
\[
	\xymatrix{
	\{\text{$C$-field on spacetime + compatible $B$-field on $\Sigma_{M5}$}\} \ar[d]^{\text{restrict $C$-field}}\\ 
	\{\text{$C$-field on $\Sigma_{M5}$ + compatible $B$-field on $\Sigma_{M5}$}\} \ar[d]^{\text{assemble into higher gauge field}}\\ 
	\{\text{$\String^{c_2}(4)$ higher gauge field on $\Sigma_{M5}$}\}
	}
\]
where the $C$-field is valued in $J$-twisted 4-cohomotopy, and the $B$-field is valued in twisted 3-cohomotopy.
In principle, this map could be neither surjective nor injective.
If it is not surjective, then there are additional constraints on the higher gauge theory coming from the spacetime topology, and the choice of M5-brane embedding.
This would tend to lead to fewer sectors, potentially removing some of the zoo of torsion phenomena in Table~\ref{fig:results}.
However, this has to be balanced against a potential \emph{increase} in sectors coming from a lack of injectivity: this would result from multiple, topologially inequivalent combinations of $C$- and $B$-fields (in this twisted cohomotopy picture) giving rise to higher gauge fields on $\Sigma_{M5}$ in the same topological sector.
This type of analysis is beyond the scope of the current article, though we note the superficially similar analysis of a forgetful map in \cite{BSS19}, related to fractional D-brane charges.

However, if one wishes to engage in a little more speculation, taking the higher gauge fields on the M5-brane as a pure exercise in field theory, with no relation to the M-theory spacetime picture in \cite{FSS_20b}, then there may be additional interpretations for some of these torsion sectors.

So far, we have not particularly discussed the group structure on the set of topological sectors.
For the case of $SU(2)$-instantons on $S^4$, we can take the sector label in $\ZZ$ to represent total change, and this is additive on combining instantons with disjoint support.
For the case of the $\String^{c_2}(4)$ higher gauge fields as calculated here, we don't just have the result that there are exactly eight topological sectors, but that they form the finite group $(\ZZ_2)^3$.
The three sectors $(1,0,0)$, $(0,1,0)$ and $(0,0,1)$ are independent, and could be interpreted as being a kind of quasi-particle excitation of the M5-brane on which they live.
The addition in the group $(\ZZ_2)^3$ corresponds to combining gauge fields with disjoint support, so that these three quasi-particles are their own antiparticles, as $(1,0,0) + (1,0,0) = (0,0,0)$, and similarly for the other two cases.
This should be interpreted in the sense that combining two copies of $(1,0,0)$ ceases to be topologically stable and can decay to the trivial solution.
It remains to be seen, though, whether such a ``higher quasi-particle'' interpretation stands up to scrutiny with a more physical mindset.

\section{Mathematical details}

\subsection{Low-dimensional homotopy groups of the classifying space}
\label{sect:pi_n_BStringc}

The results of this section use standard results in elementary homotopy theory, but we gather them here for easy reference.
A standard text for more background of this section is the first few  chapters of \cite{Switzer}.

The classifying space of $\String^{c_2}(4)$ is defined \cite{SSS} to be the homotopy pullback
\begin{equation}\label{eq:defn_Stringc2}
	\raisebox{5ex}{\xymatrix{
		B\String^{c_2}(4) \ar[r]_{\ }="s" \ar[d]_\pi & B\Sp(1)_L \ar[d]^{c_2}\\
		B\Spin(4) \ar[r]_{\frac12 p_1}^{\ }="t" & B^3U(1)
		\ar@{}"s";"t"|{\simeq}
	}}
\end{equation}
where $\frac12p_1$ and $c_2$ are the universal fractional first Pontryagin class and the universal second Chern class, respectively.\footnote{The reason for the subscript on $B\Sp(1)$ is that $\Spin(4)\simeq \Sp(1)_L \times \Sp(1)_R$, and these two copies play different roles.}
\begin{remark}
\label{rem:homotopy_pullback}
It is worth emphasising how maps to this space behave, as the construction is not common outside of algebraic topology/homotopy theory.
As discussed earlier, a map to $B\String^{c_2}(4)$ is given by a map $f$ to $B\Spin(4)$ and a map $g$ to $B\Sp(1)_L$, together with a compatibility condition.
This condition is the data of a homotopy between $f\circ \frac12p_1$ and $g\circ c_2$, both maps to $B^3U(1)\simeq K(\mathbb{Z},4)$.
Thus both of these composite maps not only give the same class in fourth integral cohomology, there is additional data relating them.
\end{remark}

We can present the map $c_2$ as a fibration without changing the homotopy type of $B\Sp(1)_L$, so that the above square can be assumed to be an honest pullback, and $\pi$ will thus also be a fibration.
Moreover, the fibre of $c_2$, and hence $\pi$, is $B\String_{\Sp(1)_L}$.
The main tool we will use is the homotopy long exact sequence of a pointed fibration $(F,p)\stackrel{i}{\into} (E,p)\xrightarrow{\pi} (B,b)$ with $\pi(p)=b$ and $F = \pi^{-1}(b)$ (e.g.~\cite[\S4.7]{Switzer}):
\[
	\cdots \to \pi_{n+1}(B,b) \to \pi_n(F,p) \xrightarrow{i_*} \pi_n(E,p) \xrightarrow{\pi_*} \pi_n(B,b) \to \pi_{n-1}(F,p) \to \cdots
\]
If we are given a map $f\colon B' \to B$, we have a pullback square of fibrations
\[
	\xymatrix{
		E' \ar[d]_{\pr_1} \ar[r]^-{\widetilde{f}} & E \ar[d]^\pi \\
		B' \ar[r]_f & B
	}
\]
such that the fibre of $E' := B' \times_B E \to B'$ is canonically isomorphic to $F$ (we shall silently make this identification in what follows).
If we are given a basepoint $b'\in B$, then this becomes a square of pointed maps, where $p' := (b',p) \in E'$, and there is then a diagram of homotopy groups
\begin{equation}\label{eq:les_natural}\hspace{-1em}
		\raisebox{5ex}{\xymatrix@C=0.6pc{
		 \cdots\ar[r] & \pi_{n+1}(B',b') \ar[r] \ar[d]_{f_*} & \pi_n(F,p) \ar[r] \ar[d]_{\simeq} & \pi_n(E',p') \ar[r] \ar[d]_{\widetilde{f}_*} & \pi_n(B',b') \ar[r] \ar[d]_{f_*} & \pi_{n-1}(F,p) \ar[r] \ar[d]_{\simeq} & \cdots\\
		 \cdots\ar[r] & \pi_{n+1}(B,b) \ar[r] & \pi_n(F,p) \ar[r] & \pi_n(E,p) \ar[r] & \pi_n(B,b) \ar[r] & \pi_{n-1}(F,p) \ar[r] & \cdots
	}}
\end{equation}
where the rows are exact.
Applying the homotopy long exact sequence to the fibrations $G \to EG \to BG$ and\footnote{$PG$ is the space of paths starting at the identity of $G$, and $\Omega G$ denotes the subspace of based loops.} $\Omega G \to PG \to G$ gives the standard isomorphisms
\begin{equation}\label{eq:standard_iso_on_htpy_grps}
	\pi_{k+1}(BG) \simeq \pi_k(G)
	\simeq \pi_{k-1}(\Omega G),
\end{equation}
as both $EG$ and $PG$ are contractible, and where these spaces have their canonical basepoints.
We also note the fact that the functors $\pi_k$, $k\geq 0$ preserve products: $\pi_k(X\times Y) \simeq \pi_k(X) \times \pi_k(Y)$ (suppressing basepoints).

The only specific homotopy groups needed below are the following:

\begin{lemma}\label{lemma:pis_of_BStringc}
The low-dimensional homotopy groups of $B\String^{c_2}(4)$ are
\[
	\pi_k(B\String^{c_2}(4)) = \begin{cases}
	0 & k=1,2,3\\
	\ZZ^2 & k=4\\
	(\ZZ_2)^3 & k=5,6\\
	\end{cases}
\]
\end{lemma}

\begin{proof}
We will examine the homotopy long exact sequence diagram \eqref{eq:les_natural} applied to the pullback square \eqref{eq:defn_Stringc2} (leaving basepoints implicit), for $n=1,\ldots,6$, and use the following standard facts:
\begin{enumerate}
\item $\pi_k(B\Spin(4)) = 0$ for $1\leq k \leq 3$;
\item $\pi_{k+1}(B\Spin(4)) \simeq \pi_k(\Spin(4))$ for all $k$, by \eqref{eq:standard_iso_on_htpy_grps};
\item The isomorphism $\Spin(4) \simeq \Sp(1)_L \times \Sp(1)_R$ means that $\pi_k(\Spin(4)) \simeq \pi_k(\Sp(1))\times \pi_k(\Sp(1))$;
\item The group $\Sp(1)$ is the unit quaternions, hence is diffeomorphic to $S^3$;
\item $\pi_4(S^3) \simeq \ZZ_2$, $\pi_5(S^3) \simeq \ZZ_2$, and $\pi_6(S^3) \simeq \ZZ_{12}$, so that $\pi_5(B\Spin(4))\simeq (\ZZ_2)^2 \simeq \pi_6(B\Spin(4))$ and $\pi_7(B\Spin(4))\simeq (\ZZ_{12})^2$;
\item $B^3U(1)$ is an Eilenberg--Mac~Lane space, namely a $K(\ZZ,4)$, so that $\pi_k(B^3U(1)) = 0$ for $k\neq 4$, and $\pi_4(B^3U(1)) = \ZZ$;
\item The defining property of $\String_G$ is that $\pi_k(B\String_G) = 0$ for $k\leq 4$, but $\pi_k(B\String_{\Sp(1)}) \simeq \pi_k(B\Sp(1))$ for $k>4$, so that $\pi_5(B\String_G) \simeq \pi_4(S^3) \simeq \ZZ_2$ and $\pi_6(B\String_G) \simeq \pi_5(S^3) \simeq \ZZ_2$.
\end{enumerate}
Here is the diagram again, omitting basepoints\footnote{All the spaces are simply-connected, so there is no basepoint-dependence in any case.} and labels on maps where not needed:
\[
	\centerline{
	\xymatrix@C=.7pc{
		 \pi_{n+1}(B\Spin(4)) \ar[r] \ar[d] & \pi_n(B\String_{\Sp(1)}) \ar[r] \ar[d]_{\simeq} & \pi_n(B\String^{c_2}(4)) \ar[r] \ar[d] & \pi_n(B\Spin(4)) \ar[r] \ar[d] & \pi_{n-1}(B\String_{\Sp(1)}) \ar[d]_{\simeq} \\
		 \pi_{n+1}(B^3U(1)) \ar[r] & \pi_n(B\String_{\Sp(1)}) \ar[r] & \pi_n(B\Sp(1)) \ar[r] & \pi_n(B^3U(1)) \ar[r] & \pi_{n-1}(B\String_{\Sp(1)})
	}}
\]
With the above observations, we get the following 
sequences and diagrams.
For $n=1,2,3,4$ we need only consider the top (exact) row.

For $n=1,2$ this becomes
\[
0 \to 0 \to \pi_n(B\String^{c_2}(4)) \to 0 \to 0\qquad (n=1,2),
\]
which gives $\pi_1(B\String^{c_2}(4))=\pi_2(B\String^{c_2}(4))=0$.

For $n=3$ we get
\[
		\ZZ^2 \to 0 \to \pi_3(B\String^{c_2}(4)) \to 0 \to 0 
\]
which again will give $\pi_3(B\String^{c_2}(4))=0$.

For $n=4$ we get
\[
	(\ZZ_2)^2 \to 0 \to \pi_4(B\String^{c_2}(4)) \to \ZZ^2 \to 0 
\]
and so $\pi_4(B\String^{c_2}(4)) \simeq \ZZ^2$.

For $n=5$ we have the commuting diagram of abelian groups
\[
	\xymatrix{
		(\ZZ_2)^2 \ar[r]^-{(a)} \ar[d] & \ZZ_2 \ar[r]^-{(b)} \ar[d]_= & \pi_5(B\String^{c_2}(4)) \ar[r] \ar[d] & (\ZZ_2)^2 \ar[r] \ar[d] & 0 \\
		0 \ar[r] & \ZZ_2 \ar[r]^-{(c)} & \ZZ_2 \ar[r] & 0 
	}
\]
where the rows are exact.
Since the leftmost square commutes, we see that the map $(a)$ is the zero map, so that by exactness of the top row, $(b)$ is injective, and we thus have a short exact sequence
\[
	0\to \ZZ_2 \to \pi_5(B\String^{c_2}(4)) \to (\ZZ_2)^2 \to 0.
\]
Exactness in the bottom row implies that the map $(c)$ is the identity map.
Our diagram has now been reduced to
\[
	\xymatrix{
		0 \ar[r] & \ZZ_2 \ar[r]^-{(b)} \ar[dr]_= & \pi_5(B\String^{c_2}(4)) \ar[r] \ar[d]^{(d)} & (\ZZ_2)^2 \ar[r] & 0 \\
		& &  \ZZ_2
	}
\]
and the commuting triangle means that $(d)$ is a left splitting of the exact sequence.
Thus by, e.g.~\cite[Lemma 3.22]{Switzer}, $\pi_5(B\String^{c_2}(4)) \simeq (\ZZ_2)^3$.

For $n=6$ we have a similar commuting diagram of abelian groups
\[
	\xymatrix{
		(\ZZ_{12})^2 \ar[r]^-{(e)} \ar[d] & \ZZ_2 \ar[r] \ar[d]_= & \pi_6(B\String^{c_2}(4)) \ar[r] \ar[d] & (\ZZ_2)^2 \ar[r] \ar[d] & 0 \\
		0 \ar[r] & \ZZ_2 \ar[r]^= & \ZZ_2 \ar[r] & 0 
	}
\]
with exact rows. The map $(e)$ is trivial, so an identical analysis to the case $n=5$ means that the group $\pi_6(B\String^{c_2}(4))$ fits into the same short exact sequence, again with a left splitting.
Thus $\pi_6(B\String^{c_2}(4))\simeq (\ZZ_2)^3$, and this completes the proof.
\end{proof}

We need two other results about $B\String^{c_2}(4)$, which follow from general results about topological groups.
In the lemma, and below, we use $LX$ to denote the space of \emph{all} continuous loops in the space $X$.
When $G$ is a topological group, then $LG$ is a topological group where multiplication is pointwise multiplication of loops.
There is an exact sequence of (continuous) homomorphisms
\[
	\Omega G \to LG \xrightarrow{\ev} G
\]
where $\ev$, a fibration, is evaluation at the basepoint of the circle, and $\Omega G = \ker(\ev)$.

\begin{lemma}
\label{lemma:LG_semidirect_prod}
The evaluation map $\ev$ has a continuous splitting $G\to LG$ including $G$ as the subgroup of constant loops, so that $LG \simeq G\ltimes \Omega G$ (in particular, \emph{as topological spaces}, is just a product).
\end{lemma}

\noindent The next result seems to be folklore, and is usually only stated for the case when $G$ is a Lie group.

\begin{lemma}
\label{lemma:BL_is_LB}
Let $G$ be a path-connected (compactly-generated weakly Hausdorff) topological group. Then there is a weak homotopy equivalence $BLG \to LBG$, in the sense that for all $k\geq 0$, $\pi_k(BLG) \simeq \pi_k(LBG)$.
\end{lemma}

\begin{proof}
Following a suggestion of Tom Goodwillie, we can in fact give a weak homotopy equivalence $\Omega LBG \sim LG$, which implies that $LBG$ is a delooping of $LG$, or in other words, $LBG$ is a model for $BLG$. 
The composite of the following sequence of natural maps gives the needed weak homotopy equivalence:
\begin{align*}
\Omega LBG & \xrightarrow{\simeq} \Hom_*(S^1_*,LBG)\\ 
			& \xrightarrow{\simeq} \Hom_*(S^1_* \wedge S^1,BG)\\ 
			& \xrightarrow{\simeq} \Hom_*(S^1 \wedge S^1_*,BG)\\ 
			& \xrightarrow{\simeq} \Hom(S^1,\Omega BG)\\ 
			& \xrightarrow{\simeq} L\Omega BG\\
			& \xrightarrow{\sim} LG\,.
\end{align*}
Here the smash product $S^1_*\wedge  S^1$ is between the pointed circle $S^1_*$ and the circle with a disjoint basepoint added, and we have used the symmetric monoidal  closed structure on the category of compactly generated weakly Hausdorff pointed spaces with smash product.
\end{proof}

We also know that despite $\String^{c_2}(4)$ being a Lie 2-group, whence the higher gauge theory under consideration, there is a (compactly generated weakly Hausdorff) topological group $\mathcal{G}$ such that $B\String^{c_2}(4) \simeq B\mathcal{G}$ \cite{BS}---for example the geometric realisation of a certain simplicial Lie group arising from a crossed module representing $\String^{c_2}(4)$. 
Abusing notation, we shall also denote this topological group by $\String^{c_2}(4)$, as we are only interested in homotopical information in this article.
The homotopy type of the 2-group and the topological group are the same, in any case, and the topological group $\String^{c_2}(4)$ will be path-conncted, by virtue of $\Spin(4)$, $\Sp(1)$ and any $K(\ZZ,3)$ space being path-connected.
Thus it makes sense to say $LB\String^{c_2}(4)$ is weakly homotopy equivalent to $BL\String^{c_2}(4)$, using Lemma~\ref{lemma:BL_is_LB}, and that $L\String^{c_2}(4) \simeq \String^{c_2}(4) \ltimes \Omega\String^{c_2}(4)$, using Lemma~\ref{lemma:LG_semidirect_prod}.
We can thus use the isomorphisms \eqref{eq:standard_iso_on_htpy_grps} and preservation of products to calculate homotopy groups of the form $\pi_k(LBG)$.

The last fact we need is that for simply-connected spaces $X$, there is an isomorphism  $[S^k,X]\simeq \pi_n(X,x)$ for any $x\in X$ \cite[\S III.1]{Whitehead}.
We shall use this result repeatedly below, with one variation in \S\ref{subsect:S3xS3_Wick}.

\subsection{Calcuating the topological sectors}
\label{sect:calc_sectors}

With these homotopical tools in hand, we can return to considering the topological sectors for a $\String^{c_2}(4)$ higher gauge theory on the worldvolume $\Sigma_{M5}$ of a single heterotic M5-brane.
Each sector calculation below works towards the point of applying the Lemma~\ref{lemma:pis_of_BStringc} as a black box.
Note also that we repeatedly use the idea that a ``vanishing at $\infty$'' boundary condition on a map is equivalent to asking that the map extends continuously to a one-point compactification. 
Or more precisely, since we only require such boundary conditions in \emph{some} directions, we are using an isomorphism of the form
\[
\{\RR^k\times Y \xrightarrow{f} B \mid \forall y\in Y,\ \lim_{|x|\to\infty}f(x,y) = b\} \simeq \{S^k\times Y \xrightarrow{\overline{f}} B,\mid \forall y\in Y,\ \overline{f}(\infty,y) = b\},
\]
where $(B,b)$ is a pointed space, and $\infty\in S^k$ denotes the added point at infinity.

\subsubsection{Trivial cases}
\label{subsect:trivial}

It is worth noting that for an M5-brane worldvolume of the form $\RR^{5,1}$, with no boundary conditions, there are no topologically non-trivial sectors, because in this case the sectors are classified by homotopy classes of arbitrary maps $\RR^{5,1} \to B\String^{c_2}(4)$, and $\RR^{5,1}$ is contractible (we recall that no metric information enters into this, or the other calculations here).

Even if we wrap around just one or two circle directions, or a 2-sphere, we can get no nontrivial sectors:
\begin{align*}
[\RR^{4,1}\times S^1,B\String^{c_2}(4)] \simeq [S^1,B\String^{c_2}(4)] &\simeq \pi_1(B\String^{c_2}(4)) = 0,\\
[\RR^{3,1}\times S^2,B\String^{c_2}(4)] \simeq [S^2,B\String^{c_2}(4)] &\simeq \pi_2(B\String^{c_2}(4)) = 0,
\end{align*}
and
\begin{align*}
[\RR^{3,1}\times\TT^2,B\String^{c_2}(4)] & \simeq [\TT^2,B\String^{c_2}(4)] \\
& \simeq [S^1,LB\String^{c_2}(4)] \\
& \simeq [S^1,BL\String^{c_2}(4)] \\
& \simeq \pi_0(L\String^{c_2}(4)) \\
& \simeq \pi_0(\String^{c_2}(4) \ltimes \Omega\String^{c_2}(4)) \\
& \simeq \pi_2(B\String^{c_2}(4)) = 0.
\end{align*}

\subsubsection{Unwrapped M5-brane, spatial boundary conditions}
\label{subsect:flat_decay_at_infty}

We will consider the case that $\Sigma = \RR^{0,1}\times \RR^5$, and will also assume that the gauge fields are stable in time, so that they do not become trivial at past or future infinity. 
This case is equivalent to $\Sigma = \RR^{0,1} \times S^5$, so that is also covered by this calculation.
Because $\RR^{0,1}$ is contractible, and we are imposing no constraints in the time direction, we can ignore this here and later, as appropriate.
We will assume that the gauge fields vanish at spatial infinity, 
and so the topological sectors are given by 
\[
	[S^5,B\String^{c_2}(4)]_* = \pi_5(B\String^{c_2}(4)) \simeq (\ZZ_2)^3.
\]
Here $[-,-]_*$ denotes based homotopy classes of based maps: those that take the canonical basepoint $\infty\in S^5$ to the basepoint $*$ in the classfying space.

\subsubsection{Wrapped around one circle, spatial boundary conditions}
\label{subsect:S1_decay}

In this case we are considering (homotopy classes of) maps $\RR^{4,1}\times S^1 \to B\String^{c_2}(4)$, such that at the limit of spatial infinity, the map approaches the basepoint $*$ in the classifying space. 
We can again ignore the factor of $\RR^{0,1}$, as there are no boundary conditions in the time direction.
The set of maps we need to consider is then isomorphic to
\[
	\{ c\colon S^4\times S^1 \to B\String^{c_2}(4) : c\big|_{\{\infty\}\times S^1} = \text{constant at }* \}
\]
The set of such maps up to homotopy is isomorphic to
\begin{align*}
	[S^4,\Omega B\String^{c_2}(4)]_* & = \pi_4(\Omega B\String^{c_2}(4))\\
	& \simeq  \pi_4(\Omega B\String^{c_2}(4)) = (\ZZ_2)^3
\end{align*}
as needed.

\subsubsection{Wick rotated, and wrapped around a 4-sphere and a 2-torus}
\label{subsect:T2}

Here, because we have Wick rotated, the set of maps we are looking at is
\begin{align*}
	[S^4\times \TT^2 , B\String^{c_2}(4)] & \simeq [S^4,L^2B\String^{c_2}(4)]\\
	& \simeq \pi_4(L^2B\String^{c_2}(4)).
\end{align*}
Now we can apply Lemmas~\ref{lemma:BL_is_LB} and \ref{lemma:LG_semidirect_prod}, to calculate that we get a composite weak homotopy equivalence
\begin{align*}
L^2B\String^{c_2}(4) & \xrightarrow{\sim} LBL\String^{c_2}(4)\\
					& \simeq LB\big(\String^{c_2}(4)\ltimes \Omega\String^{c_2}(4)\big)\\
					& \xrightarrow{\sim} BL \big(\String^{c_2}(4)\ltimes \Omega\String^{c_2}(4)\big)\\
					& \simeq B\Big\{\big(\String^{c_2}(4)\ltimes \Omega\String^{c_2}(4)\big) \ltimes \Omega\big(\String^{c_2}(4)\ltimes \Omega\String^{c_2}(4)\big)\Big\}\\
					& \simeq B\Big\{\big(\String^{c_2}(4)\ltimes \Omega\String^{c_2}(4)\big) \ltimes \big(\Omega\String^{c_2}(4)\ltimes \Omega^2\String^{c_2}(4)\big)\Big\}
\end{align*}
Thus we get an induced isomorphism 
\begin{align*}
	&\ \pi_4(L^2B\String^{c_2}(4)) \\
	&\ \simeq \pi_4 \Big [ B\Big\{\big(\String^{c_2}(4)\ltimes \Omega\String^{c_2}(4)\big) \ltimes \big(\Omega\String^{c_2}(4)\ltimes \Omega^2\String^{c_2}(4)\big)\Big\}\Big ]\\
	&\ \simeq \pi_3 \Big [\big(\String^{c_2}(4)\ltimes \Omega\String^{c_2}(4)\big) \ltimes \big(\Omega\String^{c_2}(4)\ltimes \Omega^2\String^{c_2}(4)\big)\Big ]\\
	&\ \simeq \pi_3(\String^{c_2}(4))\times \pi_4(\String^{c_2}(4)) \times \pi_4(\String^{c_2}(4)) \times \pi_5(\String^{c_2}(4))\\
	&\ \simeq  \ZZ^2 \times (\ZZ_2)^9.
\end{align*}

\subsubsection{Wrapped around a 3-sphere, no spatial boundary conditions}
\label{subsect:S3_unwrapped}

Because there are no boundary conditions in the non-compact directions, the calcuation is much simpler:
\begin{align*}
	[\RR^{2,1}\times S^3,B\String^{c_2}(4)] & \simeq [S^3,B\String^{c_2}(4)]\\
	& \simeq \pi_3(B\String^{c_2}(4)) =0.
\end{align*}

\subsubsection{Wrapped around a 3-sphere, spatial boundary conditions}
\label{subsect:S3_unwrapped_decay}

We are considering homotopy classes of maps of the form $c\colon \RR^{2,1} \times S^3\to B\String^{c_2}(4)$, but where $c(x,t,p) \to * \in B\String^{c_2}(4)$ as $|x| \to \infty$. 
The homotopy classes of such maps are isomorphic to
\begin{align*}
	[S^3,\Omega^2B\String^{c_2}(4)] & \simeq [S^3,\Omega\String^{c_2}(4)]\\
	& \simeq \pi_3(\Omega\String^{c_2}(4))\\
	& \simeq \pi_5(B\String^{c_2}(4))\\
	& \simeq (\ZZ_2)^3.
\end{align*}

\subsubsection{Wrapped around 3-sphere, Wick-rotated spacetime boundary conditions}
\label{subsect:S3xS3_Wick}

This example is similar to the one in Section~\ref{subsect:S3_unwrapped_decay}, where the classifying maps extend to the compactification, but satisfy a boundary condition.
But it also uses the fact  \cite[\S III.4.18]{Whitehead} that a topological group is a simple space---so that the $\pi_1$-action on the higher homotopy groups $\pi_k$ is trivial---and $S^3$ is simply-connected, so that the set of sectors is \cite[III.1.11--III.1.13]{Whitehead}
\begin{align*}
	[S^3,\Omega^3B\String^{c_2}(4)] & \simeq [S^3,\Omega^2\String^{c_2}(4)] \\
	& \simeq [S^3,\Omega^2\String^{c_2}(4)^\sim]/\pi_1(\Omega^2\String^{c_2}(4)) \\
	& \simeq \pi_3(\Omega^2\String^{c_2}(4))/\pi_1(\Omega^2\String^{c_2}(4)) \\
	& \simeq \pi_5(\String^{c_2}(4)) \simeq (\ZZ_2)^3,
\end{align*}
where $X^\sim$ denotes the universal covering space of $X$.

\subsubsection{Wrapped around a 4-sphere}
\label{subsect:S4}

As there are no boundary conditions, this is given by
\[
	[\RR^{1,1}\times S^4,B\String^{c_2}(4)] \simeq \pi_4(B\String^{c_2}(4)) = \ZZ^2.
\]

\subsubsection{Wrapped around a 3-sphere and a 2-torus}
\label{subsect:RxS3xT2}

There are no boundary conditions, so we can ignore the time direction, reuse the calculation of $L^2B\String^{c_2}(4)$ from Section~\ref{subsect:T2}, and get the set of sectors to be
\begin{align*}
	&\ [\RR^{0,1}\times S^3\times \TT^2,B\String^{c_2}(4)]\\
	&\ \simeq[S^3,L^2B\String^{c_2}(4)]\\
	&\ \simeq\pi_3 \Big [ B\Big\{\big(\String^{c_2}(4)\ltimes \Omega\String^{c_2}(4)\big) \ltimes \big(\Omega\String^{c_2}(4)\ltimes \Omega^2\String^{c_2}(4)\big)\Big\}\Big ] \\
	&\ \simeq\pi_2 \Big [ \big(\String^{c_2}(4)\ltimes \Omega\String^{c_2}(4)\big) \ltimes \big(\Omega\String^{c_2}(4)\ltimes \Omega^2\String^{c_2}(4)\big)\Big ] \\
	&\ \simeq \pi_2(\String^{c_2}(4))\times \pi_3(\String^{c_2}(4)) \times \pi_3(\String^{c_2}(4)) \times \pi_4(\String^{c_2}(4))\\
	&\ \simeq \ZZ^4 \times (\ZZ_2)^3.
\end{align*}

\section{Conclusion}

We have in this article calculated the possible topological sectors for $\String^{c_2}(4)$ higher gauge theories on various M5-brane topologies.
Many nontrivial sectors arise on examples of M5-brane topologies that have been considered in the literature, including those for 5d instantons/4d skyrmions for $SU(2)=Sp(1)$ Yang--Mills on M5-branes wrapped around an $S^1$.
There are a number of torsion sectors, whose physical interpretation is not yet clear.

It should be noted, however, that in the context of \cite{FSS_20b}, the higher gauge fields here are not arbitrary, but determined by other data, including the embedding of the M5-brane into spacetime.
The sector classification in this article only considers the topology of the M5-brane, and not the spacetime topology, which will potentially add additional nontrivial constraints.
This may lead to fewer sectors, or introduce more sectors due to more than one true sector on the whole M5-brane/spacetime system giving rise to the same topological sector just on the M5-brane.
Future work will also have to examine the torsion sectors that have shown up here, to see if they survive to the full sector classification arising from Hypothesis H.

\acknowledgments

The author is grateful to Urs Schreiber for discussing the physical interpretations of the results here, and supplying critical references, and to Neil Strickland for outlining on MathOverflow\footnote{\url{https://mathoverflow.net/a/354827/}} a proof of the folklore result Lemma~\ref{lemma:BL_is_LB} in full generality.
The anonymous referee highlighted a brief comment of Tom Goodwillie to Strickland's MathOverflow answer, resulting in a simple conceptual proof of Lemma~\ref{lemma:BL_is_LB}, now included in the article; thanks also to the referee for suggesting a number of improvements of exposition.

This work is supported by the Australian Research Council's Discovery Projects scheme (project number DP180100383), funded by the Australian Government.



\end{document}